\documentclass[10pt]{article}

\usepackage[nottoc,notlot,notlof]{tocbibind}
\usepackage{amsmath,amssymb,amsfonts,amsthm}
\usepackage{latexsym}
\usepackage{tikz}
\usepackage[normalem]{ulem}
\usetikzlibrary{automata,positioning}
\usetikzlibrary{chains,fit,shapes}
\usetikzlibrary{calc}
\usetikzlibrary{arrows}
\usepackage{float}
\usetikzlibrary{positioning,calc}
\usetikzlibrary{graphs}
\usetikzlibrary{graphs.standard}
\usetikzlibrary{arrows,decorations.markings}
\usepackage{multicol}
\usepackage{xcolor}
\if@thmsec
  	\newtheorem{theorem}{Theorem}[section]
\else
   \newtheorem{theorem}{Theorem}
\fi

\newtheorem{lemma}[theorem]{Lemma}
\newtheorem{prop}[theorem]{Proposition}
\newtheorem{obs}[theorem]{Observation}
\newtheorem{cor}[theorem]{Corollary}

  \newtheorem{dnt}[theorem]{Definition}

\newtheorem{exm}[theorem]{Example}

\numberwithin{equation}{section}
\title{$p$-Complete Square-free  Word-representation of Word-representable Graphs}

\author{\hspace{1cm} Biswajit Das \ and Ramesh Hariharasubramanian \\ 
{{\footnotesize d.biswajit@iitg.ac.in},\ {\footnotesize  ramesh\_h@iitg.ac.in}}\\{\footnotesize Department of Mathematics, Indian Institute of Technology Guwahati, Guwahati, Assam 781039, India}}

\begin{document}
	\maketitle
	
	\begin{abstract}
An undirected graph $G(V,E)$ is \textit{word-representable} if there exists a word $w$ over the alphabet $V$ such that two distinct letters $x$ and $y$ alternate in $w$ if and only if $x$ and $y$ are adjacent in $G$.

In this paper, we introduce the notions of \textit{$p$-complete squares} and \textit{$p$-complete square-free word-representable graphs}. Let $w \in \Sigma^{*}$ be a word over an alphabet $\Sigma$. The word $w$ contains a \textit{square} if it can be written in the form $w = s_{1}XXs_{2}$, where $s_{1}, s_{2} \in \Sigma^{*}$ and $X \in \Sigma^{+}$. For a subset $S \subseteq \Sigma$, let $w_{S}$ denote the word obtained from $w$ by deleting all letters not belonging to $S$. For an integer $p\ge 1$, the word $w$ is said to contain a \textit{$p$-complete square} if there exists a subset $S \subseteq \Sigma$ such that $w_{S}$ contains a square $XX$, where $X \in \Sigma^{+}$ and $|X| \ge p$. If no such subset exists, then $w$ is called \textit{$p$-complete square-free}.

A word-representable graph is called \textit{$p$-complete square-free word-representable} if it admits a word-representation that is $p$-complete square-free. This notion captures repetition patterns that are independent of vertex labelling and cannot be avoided by certain classes of word-representable graphs. We further define \textit{$p$-complete square-free uniform word-representations} for uniform words and investigate their structural properties. We show that for any integer $p \ge 1$, any graph admitting a $p$-complete square-free uniform word-representation forbids $K_p$ as an induced subgraph. We also show that, for arbitrary $p$, the recognition problem for $p$-complete square-free word-representable graphs is NP-hard. For small values of $p$, we obtain complete characterisations: for $p = 1$, only complete graphs admit such representations, while for $p = 2$, this holds only for complete and edgeless graphs. We prove that every $K_p$-free circle graph is $p$-complete square-free uniformly word-representable. Moreover, we show that graphs with representation number at most three admit $3$-complete square-free uniform word-representations. Finally, we present a constructive method for generating new $3$-complete square-free uniformly word-representable graphs.
	\end{abstract}

\section{Introduction}
 The theory of word-representable graphs is an up-and-coming research area. The notion of word-represent-able graphs was first introduced by Sergey Kitaev in the study of the Perkins semigroup \cite{kitaev2008word}. Word-representable graphs generalise several important graph classes, including circle graphs, comparability graphs, and $3$-colourable graphs. However, not all graphs are word-representable, and identifying such graphs remains an important problem.

 An important problem in the theory of word-representable graphs is to determine word-representations that contain or avoid prescribed patterns. 
 A pattern $\tau = \tau_1\tau_2\cdots \tau_m$ occurs in a word $w = w_1w_2\cdots w_n$, if there exists $1 \leq i_1 < i_2 <\cdots < i_m \leq n$ such that $\tau_1\tau_2\cdots \tau_m$ is order-isomorphic to $w_{i_1}w_{i_2}\cdots w_{i_m}$.
  Based on pattern-avoiding words, Jones \textit{et al.} \cite{JKPR15-EJC} introduced the notion of $u$-representable graphs, which generalises word-representable graphs.
 Let $u \in \mathbb{N}^*$ be a word. The \emph{reduction} of $u$, denoted by $\textit{red}(u)$, is the word obtained by replacing each occurrence of the $i^{th}$ smallest letter appearing in $u$ with $i$. For a word $u = u_1u_2\cdots u_j \in \mathbb{N}^*$ such that $\textit{red}(u)=u$, and a word $w = w_1w_2\cdots w_n \in \mathbb{N}^*$, the word $w$ has a $u$-match starting at position $i$ if $\textit{red}(w_i w_{i+1}\cdots w_{i+j-1}) = u$. If $u \in \{1,2\}^*$ and $\textit{red}(u)=u$, a labelled graph $G(V,E)$ with $V \subseteq \mathbb{N}$ is \emph{$u$-representable} if there exists a word $w \in V^*$ such that for all $x,y \in V$, the vertices $x$ and $y$ are adjacent if and only if the subword $w_{x,y}$ contains a $u$-match. An unlabelled graph $H$ is $u$-representable if it admits a labelling that results in a $u$-representable labelled graph.
  Word-representable graphs are $u$-representable for $u=11$. However, Kitaev \cite{Kitaev17-JGT} showed that every graph is $u$-representable whenever the length of $u$ is at least three. However, it was shown that not all graphs are $12$-representable. Jeff Remmel introduced the notion of a $k$-$11$-representable graph for a non-negative integer $k$, providing another generalisation of word-representable graphs. In this representation, an edge of a graph $G$ corresponds to at most $k$ occurrences of the consecutive pattern $11$ in a word representing $G$. Consequently, word-representable graphs are precisely the $0$-$11$-representable graphs. It has been shown that every graph is $2$-$11$-representable via a concatenation of permutations \cite{cheon2019k}, and that $1$-$11$-representations exist for the Chv\'{a}tal graph, the Mycielski graph, split graphs, and graphs whose vertex sets can be partitioned into a comparability graph and an independent set \cite{futorny2024new}.
  
The study of pattern-avoiding word-representations of word-representable graphs was introduced in \cite[Section 7.8]{kitaev2015words}. Later, Gao \textit{et al.} \cite{GKZ17-AJC} studied word-representable graphs that avoid the $132$-pattern. They showed that in the $132$-avoiding word-representant of a word-representable graph, each letter occurs at most twice; therefore, these graphs are circle graphs. Furthermore, they proved that all trees, cycles, and complete graphs admit $132$-avoiding word-representations. On the other hand, Mandelshtam \cite{Mandelshtam19-DMGT} studied the word-representable graphs that avoid the $123$-pattern. He showed that each letter appears at most twice in any $123$-avoiding word-representant of a word-representable graph, implying that these graphs are also circle graphs. Moreover, all paths, cycles, and complete graphs, but not all trees, are word-representable by $123$-avoiding words. In contrast to general word-representable graphs, the existence of both $132$-avoiding and $123$-avoiding representations depends on vertex labelling. Additionally, A. Takaoka \cite{takaoka2024forbidden} studied $12$-representable graphs that avoid patterns of length three. Vertex labelling also plays a crucial role in $12$-representations of graphs. In particular, it was shown that graphs avoiding the patterns $111$, $121$, $231$, and $321$ in their $12$-representations correspond to $12$-representable graphs, permutation graphs, trivially perfect graphs, and bipartite permutation graphs, respectively. This paper also provides forbidden pattern characterisations for other patterns, including $123$, $132$, and $211$. Forbidden pattern characterisations based on vertex orderings are useful for describing several important graph classes, including interval graphs, permutation graphs, and comparability graphs. For example, a graph $G$ is an interval graph if and only if there exists a vertex ordering such that, for any three vertices $x \prec y \prec z$, adjacency of $x$ and $z$ implies adjacency of $y$ and $z$ \cite{olariu1991optimal}. Additional significant results on forbidden-pattern characterisations can be found in \cite{Damaschke90-incollection,FH21-SIDMA} and \cite[Section 7.4]{BLS99}.
 
 In addition to permutation patterns, combinatorics on words also studies unordered patterns such as squares, cubes, overlaps, and borders. A word $u$ is said to contain a word $v$ as a factor if $u = xvy$, where $x$ and $y$ may be empty. A square (respectively, a cube) in a word consists of two (respectively, three) consecutive equal factors. A word $w \in \Sigma^*$ contains a square (respectively, a cube) if it can be written as $w = s_1XXs_2$ (respectively, $w = s_1XXXs_2$), where $s_1,s_2 \in \Sigma^*$ and $X \in \Sigma^+$.
  The concept of square-free words gained significant attention in the field of combinatorics on words after the work of Axel Thue in his paper \cite{thue1906uber}. Thue proved the existence of an infinite number of square-free words over ternary alphabets and opened up the area of combinatorics on words. The notion of a square-free word-representation in the context of word-representable graphs was introduced in the book \cite[Section 7.1.3]{kitaev2015words}. In the book, it was shown that word-representable graphs can be represented by cube-free words. Moreover, the existence of trivial square-free word-representations was established for all word-representable graphs except the empty graph on two vertices.
 In the paper \cite{das2024square}, it was shown that there exists a non-trivial square-free word-representation for each word-representable graph except the empty graph of two vertices. Therefore, every word-representable graph, except an empty graph on two vertices, is square-free word-representable. In Observation~\ref{obss1}, we prove the existence of a border-free word-representation for word-representable graphs.
 
As discussed above, pattern-avoidance notions that depend on vertex labelling provide characterisations of certain subclasses of word-representable graphs. However, for unordered patterns that do not depend on specific labellings, most word-representable graphs trivially avoid such patterns. From this, a general question arises about whether there exist label-independent patterns that are unavoidable for certain word-representable graphs but avoidable for others. Another motivation for introducing $p$-complete square-free word-representations comes from the definition of adjacency in word-representable graphs. A subword of a word $w$ is a word obtained by removing certain letters from $w$. Such subwords need not consist of consecutive letters in the original word, and different letters may appear in arbitrary positions between repeated occurrences. In a word-representation, two vertices are adjacent if and only if they alternate in the subword obtained by restricting the word to those two letters. Thus, adjacency is determined by subwords and not by consecutive factors of the word. In contrast, most unordered patterns studied in combinatorics on words, such as squares, cubes, and overlaps, are defined in terms of factors. Square-free or cube-free word-representations are known to exist for almost all classes of word-representable graphs. On the other hand, pattern-avoiding word-representations such as $132$-avoiding and $123$-avoiding representations are defined by avoiding these patterns in subwords rather than in factors, and they lead to a more restricted family of word-representable graphs. Motivated by this difference, it is natural to study repetition patterns defined on subwords. The notion of a $p$-complete square is based on this idea and leads to new structural properties of word-representable graphs.
 
 In this paper, we define a pattern that provides an affirmative answer to the query mentioned above. This pattern also generalises the concept of square patterns found within a word. As word-representable graphs avoid the simultaneous repetitive occurrences of a factor, we want to extend this square-free property to subwords.  
 We define a $p$-complete square by restricting a word to the letters of a given subword. We define the notation of $p$-complete square present in the word $w$ defined over the alphabet $\Sigma$ by restricting $w$ to any subset of $\Sigma$, where the restricted $w$ contains a square $XX$, where $X\in \Sigma^+$ and $|X|\geq p$. For example, the word $w=125783462145673818723546$ is defined on the alphabet $\Sigma=\{1,2,\ldots, 8\}$. If we restricted $w$ to the $\{2,5,7,8\}\subset \Sigma$, then the restricted word becomes $257825788725$. In this case,  $2578$ has two consecutive occurrences; therefore, $w$ contains a $4$-complete square. According to this definition, if $w$ contains a subword of length at least $p$ that occurs twice consecutively, then $w$ contains a $p$-complete square. This definition generalises the notion of a square defined on factors, as a factor is an example of a subword. In this paper, we analyse word-representable graphs, focusing on their word-representants that either contain or avoid square-free structures in subwords. 
 
  We define the notation of $p$-complete square-free word-representable graphs, where a word-representable graph $G(V,E)$ is represented by a word $w$ such that when restricting $w$ to any subset of $V$, that restricted word does not contain a square $XX$, where $X\in V^+$ and $|X|\geq p$. It is interesting because it allows us to explore whether such specific words can also represent all word-representable graphs. If this is not the case, we can determine which classes of graphs can be represented in this manner. We found that depending on the $p$ value, some word-representable graphs lose the square-free property when the word representing the graph is restricted to certain subsets of its vertices. Hence, depending on the value of $p$, the class of $p$-complete square-free word-representable graphs forms a proper subclass of word-representable graphs. Since every word-representable graph admits a $k$-uniform word-representation, and the minimum word length of many word-representable graphs remains unknown, the precise determination of letter occurrence patterns in non-uniform representations is difficult. Therefore, we restrict our analysis to uniform words when exploring $p$-complete square-free word-representations.
In this paper, we prove some of the specific properties of the $p$-complete square-free uniform words. It is interesting to discover what other graph properties are held in these graphs. 
  
  In \textit{Section \ref{sc2}}, we present the definitions, notations, and basic results used in this paper. In \textit{Section \ref{sc3}}, we formally define the concept of $p$-complete square-free word-representable graphs and $p$-complete square-free uniform word-representable graphs. We also show how to create a $p+1$-complete square-free uniform word-representable graph from an existing $p$-complete square-free uniform word-representable graph. Additionally, we analysed both $1$-complete and $2$-complete square-free uniform word-representable graphs. 
  We found that each $K_p$-free circle graph is a $p$-complete square-free uniform word-representable graph. We also show that, for arbitrary $p$, the recognition problem for $p$-complete square-free word-representable graphs is NP-hard. In \textit{Section \ref{sc4}}, we aim to classify the $3$-complete square-free uniform word-representable graphs. We show that only word-representable graphs with a representation number at most $3$ can be candidates for $3$-complete square-free uniform word-representable graphs. Finally, we introduce a method for generating more $3$-complete square-free uniform word-representable graphs from a known one.
  
  \section{Preliminaries}\label{sc2}
In this section, we briefly describe the necessary definitions, notations, and basic results on word-represent-able graphs.

  \begin{dnt}\textit{(\cite{kitaev2015words} ,Definition 3.0.3.)} Suppose that $w$ is a word and $x$ and $y$ are two distinct letters in $w$. The letters $x$ and $y$ alternate in $w$ if, after deleting all letters but the copies of $x$ and $y$ from $w$, either a word $xyxy\cdots$ (of even or odd length) or a word $yxyx\cdots$(of even or odd length) is obtained. If $x$ and $y$ do not alternate in $w$, then these letters are called non-alternating letters in $w$. 
 \end{dnt}
 
In a word $w$, if $x$ and $y$ alternate, then $w$ contains $xyxy\cdots$ or $yxyx\cdots$ (odd or even length) as a subword.

\begin{dnt}\textit{(\cite{kitaev2015words} , Definition 3.0.5.)}
 A simple graph $G(V, E)$ is \textit{word-representable} if there exists a word $w$ over the alphabet $V$ such that the two distinct letters $x$ and $y$ alternate in $w$ if and only if $x$ and $y$ are adjacent in $G$. If a word $w$ \textit{represents} $G$, then $w$ contains each letter of $V(G)$ at least once.
\end{dnt}
We denote adjacency between vertices $x$ and $y$ by $x\sim y$, and non-adjacency by $x\nsim y$.
For a word $w$, $w_{\{x_1, \cdots, x_m\}}$ denotes the word formed by removing all letters from $w$ except the letters $x_1, \ldots, x_m$. In a word $w$ representing a graph $G(V,E)$, if $w_{\{x,y\}}$ is of the form $(xy)^k$, $(yx)^k$, $(xy)^k x$, or $(yx)^k y$, then $x$ and $y$ alternate in $w$, and hence $x\sim y$ in $G$. If $x\nsim y$ in $G$, then $x$ and $y$ do not alternate in $w$ if at least one of the factors $xxy$, $yxx$, $xyy$, or $yyx$ occurs in $w_{\{x,y\}}$.
\begin{dnt}\label{unif}
    (\textit{\cite{kitaev2015words}, Definition 3.2.1.}) 
    A word $w$ is \textit{$k$-uniform word} if every letter in $w$ occurs exactly $k$ times.
\end{dnt}

\begin{dnt}(\textit{\cite{kitaev2015words}, Definition 3.2.3.)} A graph is \textit{$k$-representable} (or \textit{$k$-word-representable}) if there exists a $k$-uniform word representing it. 
\end{dnt}
\begin{theorem}(\textit{\cite{kitaev2008representable}, Theorem 7.})\label{krep}
     A graph $G$ is word-representable if and only if there exists an integer $k\ge 1$ such that $G$ is $k$-representable.
\end{theorem}
\begin{dnt} (\textit{\cite{kitaev2017comprehensive}, Definition 3})
    For a word-representable graph $G$, \textit{the representation number} is the least $k$ such that $G$ is $k$-representable. The class of graphs with representation number $k$ is denoted by $\mathcal{R}_k$.
\end{dnt}
It can be easily seen that $\mathcal{R}_1$ contains only complete graphs.
\begin{theorem}(\textit{\cite{kitaev2017comprehensive}, Theorem 6.})\label{cir} 
We have
$\mathcal{R}_2$ = \{$G$ : $G$ is a circle graph different from a complete graph\}.
\end{theorem}

\begin{prop}(\textit{\cite{kitaev2015words}, Proposition 3.2.7})\label{pr1}
		Let $w = uv$ be a $k$-uniform word representing a graph $G$, where $u$ and $v$ are two, possibly empty, words. Then, the word $w' = vu$ also represents $G$.
\end{prop}
\begin{prop}(\textit{\cite{kitaev2015words}, Proposition 3.0.15.})\label{pr2}
   Let $w = s_1xs_2xs_3$ be a word representing a graph $G$, where $s_1$, $s_2$ and $s_3$ are possibly empty words, and $s_2$ contains no $x$. Let $X$ be the set of all letters that appear only once in $s_2$. Then, possible candidates for $x$ to be adjacent in $G$ are the letters in $X$. 
\end{prop}
 
\begin{dnt}(\textit{\cite{broere2018word}, Definition 3.22.})\label{def1}
    For a $k$-uniform word $w$, the \textit{$i^{th}$ permutation}, $1\leq i\leq k$, is denoted by $p_i(w)$ where $p_i(w)$ is the permutation obtained by deleting all occurrences of each letter except its $i^{th}$ occurrence.
     
     We denote the $j^{th}$ occurrence of the letter $x$ in $w$ as $x_j$. 
 \end{dnt}
 \begin{exm}\label{exam1}
    For word $w=142513624356152643$, $p_1(w)=142536$, $p_2(w)=124356$, $p_3(w)=152643$.
 \end{exm}
 For a $k$-uniform word $w$, $p_1(w)$ and $p_k(w)$ are called the \textit{initial permutation} and \textit{final permutation} of $w$, respectively. The initial permutation is denoted by $\pi(w)$ and the final permutation is denoted by $\sigma(w)$. In Example \ref{exam1}, $\pi(w)=142536$ and $\sigma(w)=152643$. For a word $w$, $w_{\{x_1, \cdots, x_m\}}$ denotes the word after removing all letters except the letters $x_1, \ldots, x_m$ present in $w$. 
 
	\begin{obs}(\textit{\cite{kitaev2008representable}, Observation 4})\label{pw}
		Let $w$ be the word-representant of $G$. Then $\pi(w)w$ also represents $G$.
	\end{obs}
\begin{dnt}
A word $u$ contains a word $v$ as a \textit{factor} if $u = xvy$ where $x$ and $y$ can be empty words.
\end{dnt}

\begin{exm}
    The word $421231423$ contains the words $123$ and $42$ as factors, while all factors of the word $2131$ are $1$, $2$, $3$, $21$, $13$, $31$, $213$, $131$ and $2131$.
\end{exm}
\begin{dnt}(\textit{\cite{kitaev2015words},  Definition 3.4.1.})
    A graph $G$ with the vertex set $V = \{1,\ldots, n\}$ is permutationally representable if it can be represented by a word of the form $p_1\cdots p_k$, where $p_i$ is a permutation of $V$ for $1 \leq i \leq k$. If $G$ can be represented permutationally involving $k$ permutations, we say that $G$ is \textit{permutationally $k$-representable}.
\end{dnt}
\begin{theorem} (\textit{\cite{kitaev2015words}, Theorem 3.4.7.}) \label{tm}
    Let $n$ be the number of vertices in a graph $G$ and $x \in V (G)$ be a vertex of degree $n-1$ (called a dominant or all-adjacent vertex). Let $H = G\setminus {x}$ be the graph obtained from $G$ by removing $x$ and all edges incident to it. Then $G$ is word-representable if and only if $H$ is permutationally representable.
\end{theorem}
\begin{dnt}(\textit{\cite{kitaev2015words}, Definition 5.4.5.})
     A subset $X$ of the set of vertices $V$ of a graph $G$ is a module if all members of $X$ have the same set of neighbours among vertices not in $X$ (that is, among vertices in $V \setminus X$).
\end{dnt}
\begin{dnt}
    An orientation of a graph is \textit{transitive} if the presence of edges $u \rightarrow v$ and $v \rightarrow z$ implies the presence of the edge $u \rightarrow z$. An undirected graph is called a comparability graph if it admits a transitive orientation.
\end{dnt}

\begin{theorem} (\textit{\cite{kitaev2015words}, Theorem 5.4.7.}) \label{tm1}
    Suppose that $G$ is a word-representable graph and $x \in V (G)$. Let $G'$ be obtained from $G$ by replacing $x$ with a module $M$, where $M$ is any comparability graph (in particular, any clique). Then $G'$ is also word-representable.
\end{theorem}
In the following, we present the definitions of some patterns and the known results for these patterns.
 \begin{dnt}
    A \textit{square} (resp., \textit{cube}) in a word is two (resp., three) consecutive equal factors. A word $w\in \Sigma^*$ contains a square (resp.,cube) if $w=s_1XXs_2$ (resp., $w=s_1XXXs_2$), where $s_1,s_2\in \Sigma^*,X\in \Sigma^+$.
\end{dnt}
\begin{dnt}
    An \textit{overlap} is a word of the form $axaxa$, where $a \in \Sigma$, and $x \in \Sigma^*$. From this definition, we can clearly see that an overlap contains a square $axax$.
\end{dnt}
\begin{dnt}
    A word $w$ is \textit{bordered} if $w = uvu$ for some words $u$ and $v$ with $u$ non-empty.
\end{dnt}
\begin{theorem} (\textit{\cite{das2024square}, Theorem 2.}) \label{tm2}
     If $G$ is a connected graph and $w$ is a word representing $G$ where $w$ contains at least one square, then there exists a square-free word $w'$ that represents $G$.
\end{theorem}
\begin{theorem}(\textit{\cite{das2024square}, Theorem 3.})\label{tm3}
Suppose $G$ is a disconnected word-representable graph, and $G_i$, $1\leq i\leq n$, $n\in \mathbb{N}$ are the connected components of $G$. Let $w_i$ be the square-free word-representation of $G_i$, and $G_1$ be a non-empty word-representable graph. Then the word $w=w_1\setminus l(w_1)w_2\cdots w_nl(w_1)\sigma(w_n)\cdots\sigma(w_2)$ $\sigma(w_1)\setminus l(w_1)\sigma(w_2)\cdots$ $\sigma(w_n)l(w_1)$ represents $G$ and $w$ is a square-free word. Here, $l(w_1)$ denotes the last letter of $w_1$ and $w_1\setminus l(w_1)$ denotes the word obtained by removing the last letter from $w_1$. 
\end{theorem}

\begin{lemma}(\textit{\cite{das2024square}, Lemma 4.})\label{lmk}
     If $G$ is a connected word-representable graph and the representation number of $G$ is $k$, then every $k$-uniform word representing $G$ is square-free.
\end{lemma}
Since, by definition, every overlap contains a square, removing the square also eliminates the overlap.
 According to Theorems \ref{tm2} and \ref{tm3}, there exists a square-free word-representation for a word-representable graph except the empty graph of order $2$. This implies that every word-representable graph can also be represented by an overlap-free word, as the empty graph with two vertices can be represented by the word $1122$, where $1$ and $2$ represent the vertices of the empty graph. It is clear that the word $1122$ is overlap-free. Since this word is also border-free, an empty graph with two vertices can also be represented by a border-free word. In the following observation, we prove that for any word-representable graph $G$, there exists a border-free word $w$ that represents $G$.
.
\begin{obs}\label{obss1}
     If $G(V,E)$ is a word-representable graph, then there exists a border-free word $w$ that represents the graph $G$.
 \end{obs}
\begin{proof}
Let $w$ be a word that represents a word-representable graph $G$, and suppose $w$ contains a border. If $ w $ is not a uniform word (see Definition~\ref{unif}), we can transform it into a uniform word using the proof steps outlined in Theorem \ref{krep}. After this transformation, if $w = uvu$ holds, then according to Proposition \ref{pr1}, the word $ w' = uuv$ also represents the graph $ G $. Since $ w' $ contains a square, we can remove that square, as mentioned in Theorems \ref{tm2} and \ref{tm3}. Hence, every word-representable graph admits a border-free word-representation.
\end{proof}

 Recognition problem of word-representable graph and deciding the representation number of word-representable graphs are $NP$-complete problems.
	\begin{cor}(\textit{\cite{halldorsson2016semi}, Corollary 2.})\label{npc}
		The recognition problem for word-representable graphs is NP-complete.
	\end{cor}
    \begin{prop}(\textit{\cite{halldorsson2016semi},  Proposition 8.})\label{knpc}
        Deciding whether a given graph is a $k$-word-representable graph, for any given $3 \leq k \leq \lceil n/2\rceil$, is NP-complete. 
    \end{prop}
A crown graph $H_{n,n}$ is obtained from the complete bipartite graph $K_{n,n}$ by removing a perfect matching. The following theorems show the representation number of a crown graph.
\begin{theorem}(\textit{\cite{glen2018representation}, Theorem 5.})\label{crown1}
For $n \geq 1$, the representation number of a crown graph $H_{n,n}$ is at least $\lceil n/2 \rceil$.
\end{theorem}
\begin{theorem}(\textit{\cite{glen2018representation}, Theorem 7.})\label{crown2}
If $n \geq 5$, then the crown graph $H_{n,n}$ is $\lceil n/2 \rceil$-representable.
\end{theorem}
In this paper, the notation $w=a_1a_2\cdots a_n$ indicates that the word $w$ contains the factors $a_1,a_2,\ldots,a_n$, where each $a_i$ is a possibly empty word. A word $w\in \Sigma^*$ contains a square if it can be represented by $w=s_1XXs_2$ where $X\in \Sigma^+$.

\section{$p$-complete square-free  word-representation}\label{sc3}
The formal definition of $p$-complete square-free word-representable graphs is described below. All graphs considered in this paper are undirected graphs.

\begin{dnt}
Let $w \in \Sigma^{*}$ be a word and let $p$ be an integer such that $1 \le p \leq \left\lceil \frac{|w|}{2} \right\rceil +1$. The word $w$ is said to contain a \emph{$p$-complete square} if there exists a subset
$S \subseteq \Sigma$ such that the restricted word $w_S$ contains a square $XX$,
where $X \in S^{+}$ and $|X| \ge p$.
 \end{dnt}
 \begin{dnt}
Let $w \in \Sigma^{*}$ be a word and let $p$ be an integer such that $1 \le p \leq \left\lceil \frac{|w|}{2} \right\rceil +1$.
The word $w$ is called \emph{$p$-complete square-free} if there exists no subset
$S \subseteq \Sigma$ such that $w_S$ contains a $p$-complete square.
\end{dnt}
\begin{exm}
The word $w_1 = 125783462145673818725346$ is not a $3$-complete square-free word, since for the subset $S=\{2,5,7\}$, the restricted word $w_{S}=257257725$ contains a square $XX$, where $X=257$ and $|X|=3$. Moreover, $w_1$ is not a $4$-complete square-free word, since for
$S=\{2,5,7,8\}$, the restricted word $w_{S}=257825788725$ contains a square $XX$ with $X=2578$ and $|X|=4$. However, for the word $w_2 = 14213243$, there exists no subset $S$ such that the restricted word $w_S$ contains a square of the form $XX$, where $X \in \{1,2,3,4\}^{+}$ and $|X|=3$. Therefore, $w_2$ is a $3$-complete square-free word.
\end{exm}
Now, using this concept of $p$-complete square-free word, we define $p$-complete square-free word-representable graphs.
\begin{dnt}\label{def2}
 Suppose the word $w$ is a word-representant of the word-representable graph $G(V,E)$. If $w$ is a $p$-complete square-free word where $1\leq p\leq \left\lceil \frac{|w|}{2} \right\rceil$, then the graph $G$ is called a \textit{$p$-complete square-free word-representable graph} and the word $w$ is called a \textit{$p$-complete square-free word-representation} of the graph $G$.
\end{dnt}
If a word $w$ representing a graph $G(V,E)$ contains a $p$-complete square $XX$, where $X\in V^{+}$ and $p=\left\lceil \frac{|w|}{2} \right\rceil$, then $w$ must be of the form $XX$. By Theorems~\ref{tm2} and~\ref{tm3}, this square can be removed, and the resulting word still represents the graph $G$. Therefore, we do not define $p$-complete square-free word-representable graphs for the case $p=\left\lceil \frac{|w|}{2} \right\rceil$.

The graph shown in Figure \ref{fig1} can be represented with the $2$-uniform word $23123414$ and the non-uniform word $23414$. We can observe that the $2$-uniform word is $4$-complete square-free (where $w_{\{1,2,3\}}=231231$ is a square), while the non-uniform word is a $3$-complete square-free word. Additionally, $2312341$ is another non-uniform word that represents this graph; however, this word is 4-complete square-free (where $w_{\{1,2,3\}} = 231231$ is a square).

\begin{figure}[h]
\begin{center}
\begin{tikzpicture}[node distance=1cm,auto,main node/.style={circle,draw,inner sep=1pt,minimum size=5pt}]

\node[main node] (1) {1};
\node (2) [below right of=1] {};
\node (3) [ below left of=1] {};
\node[main node] (4) [below right of=2] {3};
\node[main node] (5) [below left of=3] {2};
\node[main node] (6) [right of=1] {4};

\path
(1) edge (4)
(1) edge (4)
(1) edge (5)
(1) edge (5)
(1) edge (6)
(4) edge (5);

\end{tikzpicture}

\caption{\label{fig1} Example of a graph with different $p$-complete square-free  word}
\end{center}
\end{figure}

According to Theorem \ref{krep}, every word-representable graph has a $k$-uniform word-representant. Based on this uniform word-representation, we can explore the possible $p$-complete square-free word-representations of word-representable graphs, since we know the exact number of times each letter can occur in that uniform word. However, the minimum word length for many word-representable graphs remains unknown, preventing us from determining the exact occurrence of letters in non-uniform words. Therefore, our focus is primarily on uniform words. We specifically define the notion of a $p$-complete square-free uniform word in the context of uniform words. Subsequently, we characterise $k$-uniform word-representable graphs for $k\leq 3$, in relation to $p$-complete square-free words. 
\begin{dnt}
    A word $w$ is called a $p$-complete square-free uniform word if $w$ is uniform and $p$-complete square-free.
\end{dnt}
 
\begin{dnt}\label{def3}
 Suppose the word $w$ is a uniform word-representant of the word-representable graph $G(V,E)$. If $w$ is a $p$-complete square-free uniform word where $1\le p\leq\left\lceil \frac{|w|}{2} \right\rceil$, then the graph $G$ is called a \textit{$p$-complete square-free uniform word-representable graph} and the word $w$ is called a \textit{$p$-complete square-free uniform word-representation} of the graph $G$.
\end{dnt}

\begin{dnt}
    The minimum $p$ such that a graph is $p$-complete square-free word-representable is called the graph's \textit{complete square-free representation number}. Similarly, the minimum $p$ such that a graph is $p$-complete square-free uniform word-representable is called the graph's \textit{complete square-free uniform representation number}.
\end{dnt}

Since the class of word-representable graphs is hereditary, we investigate whether the class of $p$-complete square-free word-representable graphs is hereditary as well. We provide an affirmative answer in the following lemma. 
\begin{lemma}\label{lm10}
    The class of $p$-complete square-free word-representable graphs is hereditary. 
\end{lemma}
\begin{proof}
    Suppose it is not a hereditary class. Then there exists an induced subgraph $G'(V',E')$ of the $p$-complete square-free uniform word-representable graph $G(V,E)$, such that $G'$ does not have a $p$-complete square-free uniform word-representation. However, if $w$ is a $p$-complete square-free uniform word describing $G$, then $w_{V'}$ is a $p$-complete square-free uniform word that represents $G'$, which is a contradiction.
\end{proof}

According to the definition of $p$-complete square-free uniform words, if a word-representable graph $G$ has a complete square-free uniform representation number $p$, then $G$ is also $(p+1)$-complete square-free uniform word-representable.

According to Theorems \ref{tm2} and \ref{tm3}, every word-representable graph has a square-free representation except an empty graph of two vertices. This provides an upper bound on the size of a subset $S\subseteq V$ for a word-representable graph $G(V,E)$, such that $w_{S}$ contains a $p-1$-complete square, where $w$ is a $p$-complete square-free uniform representation of the graph $G$.
\begin{obs}
    Suppose $w$ is a $p$-complete square-free uniform word-representation of the word-represe-ntable graph $G(V,E)$. Then there exists a subset $S\subseteq V$ and $|S|\leq p-1$ such that $w_{S}$ contains a square $XX$, where $X\in S^+$ and $|X|=p-1$.
\end{obs}

\begin{proof}
    Let $w$ be a $p$-complete square-free uniform word that represents the graph $G$. Suppose, for any subset $S\subseteq V$, where the word $w$ contains a $(p-1)$-complete square, the size of $S$ is more than $p-1$. We assume that the subset $S=\{a_1,a_2,\ldots, a_l\}$, where $p\leq l\leq |V|$. Without loss of generality, suppose the word $w_{S}=uXXv$, where $X=a_1\cdots a_m$, $|X|=p-1$ and $1\leq m\leq p-1$. Therefore, for the subset $S'=\{a_1,a_2,\ldots, a_m\}$, the word $w_{S'}=u'XXv'$ contains a square $XX$, where $|X|\leq p-1$. However, this contradicts our assumption.    
\end{proof}

Suppose $ G(V,E) $ is a word-representable graph with a representation number of $ k $. Let $ \mathcal{W} $ be the set containing all $ k $-uniform words that represent the graph $ G $. It is not necessarily true that if there exists a word $ w \in \mathcal{W} $ that is $ p $-complete square-free, then all words $ w' \in \mathcal{W} $ will also be $ p $-complete square-free. In the paper \cite{glen2018representation}, it was shown that the representation number of the crown graph $ H_{n,n} $ is $ \lceil n/2 \rceil $ for $ n \geq 5 $. The authors also provided a word representation for the $ H_{n,n} $ graph. According to that representation, the word $w = 1234'43'2'1'1243'34'2'1'1342'24'3'1'2341'14'3'2'$ represents the $ H_{4,4} $ graph, where the vertices are partitioned into two independent sets $ X = \{1,2,3,4\} $ and $ Y = \{1',2',3',4'\} $. It can be easily verified that $ w $ is a $ 5 $-complete square-free word. Now, if we exchange the positions of the last occurrences of $ 3' $ and $ 4' $ in the word $ w $, we can form a new word $
w' = 1234'43'2'1'1243'34'2'1'1342'24'3'1'2341'13'4'2'$. It is clear that $ w' $ still represents the $ H_{4,4} $ graph. However, this new word $ w' $ is a $ 7 $-complete square-free word. This is because the $w_{\{1,4',3'\}} = 14'3'13'4'14'3'13'4'$ contains a square $XX$, where $X=14'3'13'4'$ and $|X|= 6 $.

From this example, we can also conclude that for any word $ w \in \mathcal{W} $, if we consider a subset $ S \subseteq V $ and the word $ w_{S} $ contains a square $ XX $ where $X\in S^+$ and $ |X| = p - 1 $, the size of $ S $ may not necessarily be $ p - 1 $.

We attempt to determine whether removing a vertex from a $p$-complete square-free uniform word-representable graph transforms it into a $(p-1)$-complete square-free uniform word-representable graph. 
\begin{dnt}
   An apex vertex is a vertex that is adjacent to all other vertices in the graph $G$. It is also known as a universal vertex. 
\end{dnt} 
We now study the role of apex vertices in $p$-complete square-free uniform representations. We discovered that in a $p$-complete square-free uniform word-representable graph $G(V,E)$, if there exists an apex vertex $v$, then $v$ must be in every subset $S$ of $V$ where $w_{S}$ contains a square $XX$, where $X\in S^+$ and $|X|=p-1$. We prove this lemma below.  
\begin{lemma}\label{lmn1}
    Let $G(V,E)$ be a word-representable graph and $v\in V$ be a vertex with degree $n-1$, where $n=|V|$. If $w$ is a $p$-complete square-free uniform word representing $G$, and $p$ is the complete square-free uniform representation number, then for every $S \subseteq V$ such that $w_{S}$ contains some square $XX$, where $X\in S^+$ and $|X|=(p-1)$, $v\in S$.
\end{lemma}
\begin{proof}
Since $v$ is connected to all other vertices, according to Proposition \ref{pr2}, every other vertex must occur exactly once between two instances of $v$. Suppose there exists a set $S\subseteq V$ such that $v\notin S$, and $w_{S}$ contains a square $XX$, where $X\in S^+$ and $|X|=(p-1)$. There are two possible cases: either $X$ contains all letters of $S$, or there exists at least one letter of $S$ that does not appear in $X$. We analyse these cases separately below. 
\\\textbf{Case 1}: Let $\{a_1, a_2, \ldots, a_{p-1}\}$ be the vertices forming an $XX$-square. Without loss of generality, we assume $ X = a_1a_2\cdots a_{p-1} \quad\text{and}\quad w_S = x\, a_1a_2\cdots a_{p-1}\, a_1a_2\cdots a_{p-1}\, y.$
To obtain $w_{S\cup\{v\}}$, the vertex $v$ must occur between the two occurrences of
the subword $a_1a_2\cdots a_{p-1}$. Suppose that $v$ occurs in the $j^{th}$ position,
$1 \le j \le p$, within the subword $a_1a_2\cdots a_{p-1}$. Then, in order to satisfy
the alternation condition, the second occurrence of
$a_1a_2\cdots a_{p-1}$ must also contain $v$ in the same $j^{th}$ position.
Consequently, the word $w_{S\cup\{v\}}$ contains a factor of the form $X_1X_1$,
where $X_1$ is obtained from $a_1a_2\cdots a_{p-1}$ by inserting $v$, and hence
$|X_1| = p$. This creates a square of length $p$, which contradicts our assumption.
\\\textbf{Case 2}: Let $a_1,a_2,\ldots,a_j$ be the vertices that form the square $XX$, where $j<p-1$. Since $|X|=p-1$ and $j<p-1$, at least one vertex must occur twice in $X$. Let $a_i$ be such a vertex. Then the word $w_S$ can be written as
$
w_S = x\, a_1a_2\cdots a_i \cdots a_i \cdots a_j \, a_1a_2\cdots a_i \cdots a_i \cdots a_j \, y.
$
Since $a_i$ is adjacent to $v$, there must be exactly one occurrence of $v$ between any two consecutive occurrences of $a_i$. Hence, in the word $w_{S\cup\{v\}}$, we can place $v$ between the first and second occurrences of $a_i$, and also between the third and fourth occurrences of $a_i$, obtaining
$
w_{S\cup\{v\}} =
x\, a_1a_2\cdots a_i \cdots v \cdots a_i \cdots a_j \,
a_1a_2\cdots a_i \cdots v \cdots a_i \cdots a_j \, y.
$
However, to satisfy the alternation condition, the vertex $v$ must occur between the second and third occurrences of $a_i$. Moreover, between two consecutive occurrences of $v$, each vertex in $\{a_1,a_2,\ldots,a_j\}$ must appear exactly once.  Thus, in the word
$w_{S\cup\{v\}} =
x\, a_1a_2\cdots a_i \cdots v \cdots a_i \cdots a_j \,
a_1a_2\cdots a_i \cdots v \cdots a_i $ $\cdots a_j \, y,
$
there is already one occurrence of $v$ inside each copy of $X$. If we insert an additional occurrence of $v$ into the first copy $a_1a_2\cdots a_i \cdots a_i \cdots a_j$, then the second copy must also contain $v$ in the same position in order to preserve alternation. Consequently, this creates a square of length $|X|+2$, which contradicts our assumption.

Therefore, $\forall S \subseteq V$, such that $w_{S}$ contains square $XX$, $|X|= p-1$, $v\in S$.
\end{proof}

In a $p$-complete square-free uniform word-representable graph $G(V,E)$, the apex vertex is present in every subset $S$ of the vertices that contains a square $XX$, where $X\in S^+$ and $|X|=p-1$. Thus, removing the apex vertex creates a $(p-1)$-complete square-free uniform word-representable graph. We can prove this statement using the theorem below. 

\begin{theorem}
    Let $G(V,E)$ be a word-representable graph and $v\in V$ be a vertex with degree $n-1$ where $n=|V|$. If $G$ is a $p$-complete square-free uniform word-representable graph and $p$ is the complete square-free uniform representation number, then $G'=G\setminus v$ is a $(p-1)$-complete square-free uniform word-representable graph.
\end{theorem}

\begin{proof}
Suppose $w=P_{11}vP_{12}P_{21}vP_{22}\cdots P_{k1}vP_{k2}$ represents $G$ where $P_{ij}\in (V\setminus \{v\})^*$ and $1\leq i\leq k$, $1\leq j\leq k$. As $v$ is connected to all other vertices, according to Proposition \ref{pr2}, every other vertex should occur exactly once between two $v$'s. Therefore, every $P_{i2}P_{(i+1)1}$, $1\leq i<k$, contains all the other vertices. Also, according to Proposition \ref{pr1}, $vP_{k2}P_{11}vP_{12}P_{21}vP_{22}$ $\cdots P_{k1}$ represents $G$. So, $P_{k2}P_{11}$ also contains every other vertex exactly once. 

We know that $w'=P_{11}P_{12}P_{21}P_{22}\cdots P_{k1}P_{k2}$ represents the graph $G\setminus v$. Suppose $S\subseteq V(G')$ such that $w'_{S}$ contains a square $XX$, where $X\in V(G')^+$ and $|X|=(p-1)$. Let $w'_{S}=xa_1a_2\cdots a_{p-1}a_1a_2$ $\cdots a_{p-1}y$. So, for $w$, in $w_{S}$, $a_1a_2\cdots a_{p-1}a_1a_2$ $\cdots a_{p-1}$ is a square and $|a_1a_2\cdots a_{p-1}|=p-1$. However, according to Lemma \ref{lmn1}, $v$ should be present in $S$, which is not possible. Therefore, $w'$ does not contain any square $XX$, $|X|=p-1$. Hence, $w'$ is a $(p-1)$-complete square-free representation of $G\setminus v$.
\end{proof}

According to Lemma \ref{lmn1}, in a $p$-complete square-free uniform word-representable graph $G(V,E)$, an apex vertex can appear in any subset $S\subseteq V$ where $w_{S}$ contains a square $XX$, $X\in S^+$ and $|X|=p-1$. Here, $w$ represents the $p$-complete square-free uniform word of $G$. It is straightforward to prove that all $p$-complete square-free uniform words for the graph $G(V,E)$ contain the apex vertex in every subset of $V$, where the word is restricted to that subset containing a square $XX$, where $X\in V^+$ and $|X|=p-1$. Based on this, we introduce the notation of a $p$-complete square vertex and $p$-complete square vertex set as follows.
\begin{dnt}
   Let $G(V, E)$ be a $p$-complete square-free uniform word-representable graph, and let $ \mathbb{W} $ be the set of all $ p $-complete square-free uniform words that represent $ G $. For each vertex $ v \in V $, if for every word $ w_i \in \mathbb{W} $ there exists a subset $ S_i \subseteq V $ such that $ v \in S_i $ and the word $ w_{i_{S_i}} $ contains a square $XX$, where $X\in S_i^+$ and $|X|= p-1 $, then $ v $ is called a $ p $-complete square vertex. 

Furthermore, if there exists a subset $ S \subseteq V$, and every vertex $v\in S$ is a $p$-complete square vertex, then $S$ is called a $ p$-complete square vertex set.
\end{dnt}

\begin{exm}
    For example, according to Lemma \ref{lmn1}, for a $p$-complete square-free uniform word-represent-able graph $G(V,E)$, an apex vertex $v\in V$ presents in every set $S$, where $w_{S}$ contains a square $XX$, $|X|=p-1$. This vertex $v$ is an example of a $p$-complete square vertex.
\end{exm}
Given a $p$-complete square vertex, we present a method for constructing a $(p+1)$-complete square-free uniform word-representable graph from a known $p$-complete square-free uniform word-representable graph. We replace a $p$-complete square vertex with a $K_2$ module as described in the following theorem.

\begin{theorem}
Let the graph $G(V,E)$ be a $p$-complete square-free uniform word-representable graph where $p>2$, and $p$ is the complete square-free uniform representation number. Let the graph $G'(V',E')$ be obtained from $G$ by replacing the vertex $v\in V$ with the module $K_2$. The graph $G'$ is $(p+1)$-complete square-free uniform word-representable if and only if $v$ is a $p$-complete square vertex.
\end{theorem}
\begin{proof}
    Suppose for the graph $G(V,E)$, $w=s_1vs_2vs_3\cdots s_{k-1}vs_k$ is a $p$-complete square-free uniform word representation where $v$ is a $p$-complete square vertex. According to Theorem \ref{tm1}, $w'=s_112s_212s_3\cdots s_{k-1}$ $12s_k$, where $1$ and $2$ are the vertices of $K_2$, represents the graph $G'$. Let $S\subseteq V$ such that $w_{S}$ contains a square $XX$, where $X\in S^+$ and $|X|=p-1$. Since, $v$ is a $p$-complete square vertex,  $v\in S$ and $w_{S}=xX_1vX_2X_1vX_2y$, where $XX=X_1vX_2X_1vX_2$ and $|X_1vX_2|=p-1$. Therefore, in $w'_{(S\setminus \{v\})\cup\{1,2\}}$, there exists a square $X_112X_2X_112X_2$ where $|X_112X_2|=p$.
    
    Now, we need to prove that for any subset $S'\subseteq V'$, in the word $w_{S'}$, there does not exist a square $XX$, where $X\in (S')^+$ and $|X|>p$. Suppose for a subset $S'\subseteq V'$, the word $w_{S'}$ contains a square $XX$, where $X\in (S')^+$ and $|X|>p$. We consider the following possible cases.
    \\\textbf{Case 1}: If $1,2\in S'$, we can assume that $w'_{S'}=xX_112X_2X_1v12X_2y$. Then, the word $w'_{S'}$ contain a square $X_112X_2X_112X_2$ and $|X_112X_2|=p+1$. However, $w_{S'\cup \{v\}}=xX_1vX_2X_1vX_2y$ contains a square $X_1vX_2X_1vX_2$ where $|X_1vX_2|=p$, which contradicts our assumption. Similarly, if only $1\in S'$ or $2\in S'$, we can apply the same argument and obtain a contradiction.
    \\\textbf{Case 2}: If, $1,2\notin S'$, then let $w'_{S'}=xXXy$, where $XX$ is a square, $X\in (S')^+$ and $|X|=p+1$. But, $\forall u\in S'$, $u$ is also a vertex of $G$. So, $w_{S'}=xXXy$ contain a square $XX$, $|X|=p+1$. It contradicts our assumption.
    
    Therefore, $w'$ is a $(p+1)$-complete square-free uniform word-representation of the graph $G'$. 
    Now, we need to prove that there does not exist a $p$-complete square-free uniform word-representation of $G'$. Suppose that there exists a word $w'$, which is a $p$-complete square-free uniform word-representation of $G'$. Let $S'\subseteq V'$, such that $w'_{S'}$ contains a square $XX$, where $X\in (S')^+$ and $|X|=(p-1)$. We consider the following possible cases for $S'$.
    \\\textbf{Case 1}:  If $1,2\in S'$, then by replacing every occurrence of $1$ with $v$ and removing every occurrence of $2$ from the word $w'$, we obtain a new word $w''$ for graph $G$. However, for any $S\subseteq V$, where $v$ is an element of $S$, $w''_{S}$ contains a square $XX$, where $X\in S^+$ and $|X|=p-1$. Therefore, $w'_{(S\setminus \{v\})\cup\{1,2\}}$ must contain a square $XX$, where $X\in ((S\setminus \{v\})\cup\{1,2\})^+$ and $|X|=p$, which contradicts our assumption. If either $1$ or $2$ is in $S'$, we can create the same word $w''$ as before and use the same argument to find a contradiction.
    \\\textbf{Case 2}: If $1,2\notin S'$, then by replacing every occurrence of $1$ with $v$ and removing every occurrence of $2$ in the word $w'$, we obtain a word $w''$ for graph $G$. Using the same argument in Case 1, we can find the contradiction.

    Suppose $v$ is not a $p$-complete square vertex. So, there exists a word $w=s_1vs_2\cdots v s_{k-1}vs_k$, which is $p$-complete square-free uniform word-representation of $G$ and $\forall S\subseteq V(G)$ such that $w_S$ contains a square $XX$, where $X\in S^+$, $|X|=p-1$ and $v\notin S$. Then, according to Theorem \ref{tm1}, $w'=s_112s_2\cdots 12 s_{k-1}12s_k$, represent the graph $G'$. Suppose there exists  $S'\subseteq V(G')$, such that $w_S'$ contains a square $XX$, where $X\in (S')^+$ and $|X|=p$. We consider the following cases for $S'$.
    \\\textbf{Case 1}: If ${1,2}\in S'$, then by replacing every occurrence of $1$ with $v$ and removing every occurrence of $2$ in the word $w'$, we obtain the word $w$ for graph $G$. The word $w_{(S'\setminus \{1,2\})\cup\{v\}}$ has a square $XX$, where $X\in ((S'\setminus \{1,2\})\cup\{v\})^+$ and $|X|=p-1$. However, it is not possible because $v$ does not belong to any $S\subseteq V(G)$, such that $S$ contains a square $XX$, where $X\in S^+$ and $|X|=p-1$. If $1\in S'$ or ${2}\in S'$, then from the same process, we can obtain the word $w$ and using the same argument, we can find the contradiction on vertex $v$.
    \\\textbf{Case 2}: If ${1,2}\notin S'$, then $S'\subseteq V$. Therefore, $w_{S'}$ contains a square $XX$, where $X\in (S')^+$ and $|X|=p$. But it is not possible because $w$ is a $p$-complete square-free word.
    \\ Therefore, if $v$ is not a $p$-complete square vertex, then $G'$ is $p$-complete square-free uniform word representable.   
\end{proof}

Next, we derive a forbidden induced subgraph characterisation. We want to find the word-representable graphs, which are $p$-complete square-free uniform word-representable. We found out that $p$-complete square-free uniform word-representable graphs avoid $K_p$ as an induced subgraph. We prove this statement in the following theorem.

\begin{theorem} \label{lm11}
    If $G(V,E)$ is a $p$-complete square-free uniform word-representable graph, where $p\geq 2$, then $G$ forbids $K_p$ as an induced subgraph.  
\end{theorem}
\begin{proof}
    Suppose that there exists a $k$-uniform word $w$ that represents $G$, which is a $p$-complete square-free uniform word with $K_p$ as an induced subgraph. Let, $\{a_1,a_2,a_3,\cdots,a_p\}$ be the vertices of the graph $K_p$. So, $a_1,a_2,a_3,\ldots,a_p$ is alternating with each other in $w$. 
    Without loss of generality, we assume $\pi(w)_{\{a_1,a_2,\cdots,a_p\}}$ $=a_1a_2a_3\cdots a_p$. Therefore, we obtain that $w_{\{a_1,a_2,a_3,\cdots,a_p\}}=(a_1a_2a_3\cdots a_p)^k$.
    So, it contains a square $a_1a_2a_3\cdots a_pa_1a_2a_3\cdots a_p$ of size $p$. This contradicts our assumption, since according to the definition of the $p$-complete square-free uniform word, a square $XX$, where $|X|\geq p$, cannot occur. Therefore, $G$ forbids $K_p$ as an induced subgraph.  
\end{proof}

  However, the converse of this theorem is not true. Later, in Theorem \ref{tm7}, we provide an example that contradicts the converse statement of Theorem \ref{lm11}.

  According to Theorem \ref{tm8}, a word-representable graph can be $p$-complete square-free uniform word-representable based on $p$-value. At first, we focus on identifying the $p$-complete square-free uniform word-representable graphs when $p=1$. In the following lemma, we prove that only complete graphs have such a word-representation.
\begin{lemma} \label{lm15}  
   A graph $G(V,E)$ is $1$-complete square-free word-representable if and only if $G$ is a complete graph.
\end{lemma}
\begin{proof}
    According to the Definition~\ref{def3}, any subword of $w$ representing $K_n$ should not have any square. As the representation number of $K_n$ is $1$, the $1$-uniform word has no square. Therefore, $K_n$ is $1$-complete square-free.

    Let $G(V,E)$ be a word-representable graph that is not complete, and $w$ is a $1$-complete square-free word representing $G$. Since $G$ is not complete, $w$ is at least $2$-uniform. Since $G$ is not complete, let $x,y\in V$ such that $x\nsim y$, so from $\{xxy,yxx,yyx,xyy\}$ at least one of the factor is present in $w_{\{x,y\}}$. But in all of the factors, there exists a square $XX=xx$ or $yy$ and $|X|=1$, which is a contradiction. Therefore, $G$ does not have any $1$-complete square-free word-representation.
\end{proof}

From Lemma \ref{lm15}, we know that only the complete graph is $1$-complete square-free representable. Therefore, the complete graph is a $p$-complete square-free uniform word-representable graph for $p>1$. Hence, in the following discussions, all the graphs we considered are not complete graphs.

A $2$-complete square-free uniform word-representable graph does not contain $K_2$ as a subgraph. The only $K_2$-free graph is an empty graph, which is $2$-complete square-free uniform word-representable. We prove this statement below.
\begin{cor}\label{cor1}
    If the word-representable graph $G(V,E)$ is an empty graph, then $G$ is a $2$-complete square-free uniform word-representable graph. 
\end{cor}
\begin{proof} 
    According to Theorem \ref{lm11}, a $2$-complete square-free uniform word is $ K_2$-free. Empty graphs are the graphs that are $K_2$-free. Suppose, $G(V,E)$ is an empty graph on $n$ vertices, where $V=\{1,2,\ldots,n\}$. Then, $w=12\cdots n n\cdots 21$ represents the empty graph. We can  clearly see that any subset $S\subseteq V$ where $S=\{a_1,a_2,\ldots,a_i\}$, $a_1<a_2<\ldots <a_i$ and $1\leq i\leq n$, $w_S=a_1a_2\cdots a_ia_i\cdots a_2a_1$. The word $w_S$ contains a square $XX$, where $X=a_i$ and $|X|=1$. Therefore, the word $w$ is a $2$-complete square-free uniform word-representation of $G$.
\end{proof}

From Theorem \ref{lm11}, we obtain the following result.
\begin{cor}\label{cr1}
    If $G(V,E)$ is $3$-complete square-free uniform word-representable graph $G(V,E)$, then $G$ is $K_3$-free.
\end{cor}
\begin{proof}
   It can be seen directly from Theorem \ref{lm11}. 
\end{proof}

Since an empty graph is only $2$-complete square-free uniform word-representable, we will check whether the other $2$-representable graphs are $p$-complete square-free uniform.
For a $K_p$-free circle graph, there exists a $p$-complete square-free uniform word-representation.
Using Theorem \ref{cir}, we prove the following theorem. 

\begin{theorem}\label{lm12}
    If $G(V,E)$ is a $K_p$-free circle graph, then $G$ is a $p$-complete square-free uniform word-representable graph.
\end{theorem}
\begin{proof}
    According to Theorem \ref{cir} and Lemma \ref{lmk}, a $K_p$-free circle graph $G(V,E)$ has a $2$-uniform square-free word-representation $w$. By definition of a $p$-complete square-free uniform word, suppose that there exists a subset $S \subseteq V$ such that $w_S$ contains a square $XX$, where $X \in S^+$ and $|X| \ge p$. Let $S=\{1,2,\ldots ,k\}$, where $k\geq p$. Then, we obtain that $w_{\{1,2,\ldots ,k\}}=uXXv$, where $u,v, X \in S^*$ and $|X|\geq p$. Since $w$ is $2$-uniform, if $X$ contains a letter $x \in S$ twice, then $x$ cannot occur in the other copy of $X$. However, this is not possible, so every letter present in $X$ occurs only once. Let, $\{a_1,a_2,a_3,\cdots a_l\}$, $p\leq l\leq k$, and $P$ is the permutation of $\{a_1,a_2,a_3,\cdots ,a_l\}$ present in $X$. Without loss of generality, we assume $P=a_1a_2a_3\cdots a_l$, then $w_{\{1,2,\ldots ,k\}}=ua_1a_2a_3\cdots a_la_1a_2a_3\cdots a_lv$. Since every pair $a_i,a_j \in {a_1,a_2,\ldots,a_l}$ with $i \neq j$ alternates in $w$, therefore, $a_1, a_2, \ldots a_p$ form a $K_p$. This contradicts our assumption. Therefore, $w$ is a $p$-complete square-free uniform word. 
\end{proof}

Suppose $w$ is a $p$-complete square-free uniform word-representation of a graph $G$. 
By the definition of a $p$-complete square-free uniform word-representable graph, the value of $p$ is at most $\left\lceil \frac{|w|}{2} \right\rceil$. 
In the following theorem, we show that for a word-representable graph $G(V,E)$ with representation number $k \ge 3$, the value of $p$ is at most $\left\lceil \frac{kn}{2} \right\rceil - 1$, where $|V| = n$. This provides a tighter upper bound on possible values of $p$.

 \begin{theorem}\label{tm8}
     Suppose $G(V,E)$ is a word-representable graph with representation number $k$. Then, $G$ is $\left\lceil\frac{kn}{2}\right\rceil-1$-complete square-free uniform word-representable, where $|V|=n$ except the circle graph having $K_{n-1}$ as induced subgraph.
 \end{theorem}
 \begin{proof}
We first show that a circle graph containing $K_{n-1}$ as an induced subgraph cannot admit a
$\left\lceil \frac{kn}{2} \right\rceil - 1$-complete square-free uniform word-representation.
The representation number of any circle graph is $2$.
Hence, in this case, $ \left\lceil \frac{kn}{2} \right\rceil - 1 = n - 1$.
By Theorem~\ref{lm11}, a graph that contains $K_p$ as an induced subgraph cannot be
$p$-complete square-free uniformly word-representable.
Therefore, a circle graph containing $K_{n-1}$ cannot be $\left\lceil \frac{kn}{2} \right\rceil - 1$-complete square-free uniformly word-representable.

Now, we prove $p$-complete square-free uniform word-representation for other word-representable graph. According to Lemma~\ref{lm15}, every complete graph is $1$-complete square-free uniformly word-representable, and hence also $\left\lceil \frac{kn}{2} \right\rceil - 1$-complete square-free uniformly word-representable. Similarly, by Corollary~\ref{cor1}, every empty graph is $2$-complete square-free uniformly word-representable, and therefore, also $\left\lceil \frac{kn}{2} \right\rceil - 1$-complete square-free uniformly word-representable.

Moreover, according to Theorem~\ref{lm12}, every $K_p$-free circle graph is $p$-complete square-free uniformly word-representable. Since circle graphs have representation number $2$, the only circle graphs that fail to admit a $\left\lceil \frac{2n}{2} \right\rceil - 1$-complete square-free uniform word-representation are those containing $K_{n-1}$ as an induced subgraph.

 Finally, suppose that $w$ is a $k$-uniform word-representation of $G$, where $k \ge 3$,
and that $w$ is not $\left\lceil \frac{kn}{2} \right\rceil - 1$-complete square-free.
Then there exists a subset $S \subset V$ such that $w_S = s_1 X X s_2$, where $|X| = \left\lceil \frac{kn}{2} \right\rceil$. Since each letter of $S$ occurs exactly $k$ times in $w_S$, we have $2|X| \le k|S|$. As $|S|<n$, there exists a vertex $v \in V \setminus S$.
Because $k \ge 3$, we obtain $k|S| < kn - 2$, which implies $ |X| < \frac{kn}{2} - 1$. But, this contradicts our assumption on $|X|$. Hence, $w$ must be $\left\lceil \frac{kn}{2} \right\rceil - 1$-complete square-free uniform word. 
 \end{proof}

From Theorem~\ref{tm8}, we obtain the following corollary.
\begin{cor}\label{cor2}
Let $G$ be a word-representable graph, and let $w$ be an $l$-uniform word that represents $G$.
If the representation number of $G$ is $k$ and $l>k$, then $G$ is
$\left\lceil \frac{ln}{2} \right\rceil - 1$-complete square-free uniformly word-representable,
where $|V(G)|=n$.
\end{cor}
\begin{proof}
 Since the representation number of $G$ is $k$, by Theorem~\ref{tm8}, the graph $G$ is $\left\lceil \frac{kn}{2} \right\rceil - 1$-complete square-free uniformly word-representable, except in the case where $G$ is a circle graph containing $K_{n-1}$ as an induced subgraph. Moreover, since $k<l$, it follows that such graphs are also $\left\lceil \frac{ln}{2} \right\rceil - 1$-complete square-free uniformly word-representable. Let $G_1(V,E)$ be a circle graph, where $
V=\{a_1,a_2,\ldots,a_i,b_1,b_2,$ $\ldots,b_j,v\}$, $i+j=n-1$, the vertices $a_1,a_2,\ldots,a_i,b_1,b_2,\ldots,b_j$ induce a $K_{n-1}$, and $N(v)=\{a_1,a_2,$ $\ldots,a_i\}$. The word $w=a_1a_2\cdots a_i v b_1b_2\cdots b_j a_1a_2\cdots a_i b_1b_2\cdots b_j v a_1a_2\cdots a_i v b_1$ $b_2\cdots b_j$ represents the graph $G_1$. Restricting $w$ to vertex set $\{a_1,a_2,\ldots,a_i,b_1,b_2,\ldots,b_j\}$, , we obtain $ w_{\{a_1,\ldots,a_i,b_1,\ldots,b_j\}}= (a_1a_2\cdots a_i\, b_1b_2\cdots b_j)^3$, which contains a square $XX$, where
$ X = a_1a_2\cdots a_i\, b_1b_2$ $\cdots b_j$ and $|X|=n-1$.
 It follows that no $p$-complete square can occur for $p>n-1$.
Since $\left\lceil \frac{3n}{2} \right\rceil - 1 > n-1$, the graph $G_1$
is $\left\lceil \frac{ln}{2} \right\rceil - 1$-complete square-free uniformly word-representable
for all $l>k$. Therefore, every word-representable graph is also $\left\lceil \frac{ln}{2} \right\rceil - 1$-complete square-free uniform word-representable.
\end{proof}

\begin{cor}\label{cor3}
A graph $G$ is word-representable if and only if it admits a $p$-complete square-free uniform word-representation for some $p$.
\end{cor}
\begin{proof}
If $G$ is word-representable and $w$ is an $l$-uniform word representing $G$,
then by Corollary~\ref{cor2}, $G$ admits a
$\left\lceil \frac{ln}{2} \right\rceil - 1$-complete square-free uniform word-representation.

Conversely, if $G$ admits a $p$-complete square-free uniform word-representation, then by Definition~\ref{def3}, $G$ is word-representable.
\end{proof}

\begin{cor}
The recognition problem for $p$-complete square-free word-representable graphs,
for arbitrary values of $p$, is NP-hard.
\end{cor}
\begin{proof}
This follows directly from Corollary~\ref{npc} and Corollary~\ref{cor3}.
\end{proof}
 
 Theorem~\ref{tm8} provides an upper bound on the value of $p$ in terms of the representation number of a graph. However, this bound need not be tight, and the obtained value of $p$ may be larger than the complete square-free uniform representation number. In the next section, we study the relationship between the representation number and the complete square-free uniform representation number. In particular, we establish this connection for complete square-free uniform representation numbers up to $3$.

\section{$3$-complete square-free uniform word} \label{sc4}

According to Theorem \ref{lm12}, all triangle-free ($K_3-free$) circle graphs are $3$-complete square-free uniform word-representable. Every $2$-uniform square-free word representing a $K_3$-free graph $G$ is also a $3$-complete square-free uniform word. We now examine whether every $K_3$-free word-representable graph with representation number $k \ge 3$ admits a $3$-complete square-free uniform word-representation. In the following theorem, we prove that a $3$-complete square-free uniform word-representation does not exist for a graph with representation number $k>3$.
\begin{theorem}\label{tm7}
    If $G(V,E)$ is a word-representable graph having representation number $k>3$, then $G$ is not $3$-complete square-free uniform word-representable graph.
\end{theorem}
\begin{proof}
    Suppose that there exists a $3$-complete square-free uniform word $w$ representing the graph $G(V,E)$. Since the representation number of the graph $G$ is $k>3$, every letter occurs at least $4$ times. Since $G$ is a connected graph, there exist three vertices $a,b,c$ such that $a\sim b$, $a \sim c$ and $b\nsim c$ (otherwise, if $b\sim c$ then $abc$ would form a $K_3$). Now, in the word $w_{\{a,b,c\}}$, the possible initial permutation must be one of these six permutations: $abc$, $acb$, $bac$, $cab$, $bca$, $cba$. We discuss every case in the following:
    \\\textbf{Case 1}: If $w_{\{a,b,c\}}=abcs_1$, then $s_1$ needs to begin with $a$; otherwise, either $b$ or $c$ would occur twice between two $a$'s, removing the alternation of $a$ with $b$ or $c$. Moreover, in $s_1$, the $a$ must be followed by $cb$; otherwise, $bc$ would create the square $abcabc$. Thus, we obtain $w_{\{a,b,c\}}=abcacbs_2$. By the same reasoning applied to $s_1$, $s_2$ also starts with $a$ and is followed by $bc$; otherwise, $cb$ would create a square. Therefore, we obtain $w_{\{a,b,c\}}=abcacbabcs_3$. If $s_3=abc$ or $s_3=acb$ then $w_{\{a,b,c\}}=abcacbabcabc$ or $w_{\{a,b,c\}}=abcacbabcacb$, respectively, but either case there exists a square in $w_{\{a,b,c\}}$.
    \\\textbf{Case 2}: If $w_{\{a,b,c\}}=acbs_1$, then applying the same argument as in Case 1, $s_1$ needs to start with $a$. In $s_1$, this $a$ must be followed by $bc$; otherwise, $cb$ would create the square $acbacb$. Thus, we obtain $w_{\{a,b,c\}}=acbabcs_2$. By the same argument applied to $s_1$, $s_2$ also starts with $a$ and is followed by $cb$; otherwise, $bc$ would create a square. Therefore, we obtain $w_{\{a,b,c\}}=acbabcacbs_3$. If $s_3=abc$ or $s_3=acb$ then $w_{\{a,b,c\}}=acbabcacbabc$ or $w_{\{a,b,c\}}=acbabcacbacb$, respectively, but either case there exists a square in $w_{\{a,b,c\}}$.    
    \\\textbf{Case 3}: If $w_{\{a,b,c\}}=bacs_1$, then $s_1$ must start with $b$. If it starts with $a$, then $a$ and $b$ do not alternate, and if it starts with $c$, then $a$ and $c$ do not alternate. In $s_1$, the $b$ should be followed by $ac$; otherwise, $ca$ would remove the alternation between $a$ and $c$. But then $bacbac$ forms a square. However, in $w_{\{a,b,c\}}$, before the $1^{st}$ occurrence of $c$, $b$ can occur twice. Thus, if $w_{\{a,b,c\}}=babcs_1$, then the factor $abc$ occurs in $w_{\{a,b,c\}}$. Therefore, it follows the same condition as Case 1, except that the last occurrence of $b$ is removed because $b$ already occurs as the first letter in $w_{\{a,b,c\}}$. Then $w_{\{a,b,c\}}=babcacbabcac$, but there exists a square $XX$, where $X=babcac$ in $w_{\{a,b,c\}}$. 
    \\\textbf{Case 4}: If $w_{\{a,b,c\}}=cabs_1$, then $s_1$ must start with $c$. If it starts with $a$, then $a$ and $c$ do not alternate, and if it starts with $b$, then $a$ and $b$ do not alternate. In $s_1$, the $c$ should be followed by $ab$; otherwise, $ba$ would remove the alternation between $a$ and $b$. But then $cabcab$ forms a square. However, in $w_{\{a,b,c\}}$, before the $1^{st}$ occurrence of $b$, two $c$ can occur. Thus, if $w_{\{a,b,c\}}=cacbs_1$, then the factor $acb$ occurs in $w_{\{a,b,c\}}$. Therefore, it follows the same condition as Case 2, except that the last occurrence of $c$ is removed because $c$ already occurs as the first letter in $w_{\{a,b,c\}}$. Then $w_{\{a,b,c\}}=cacbabcacbab$, but there exists a square $XX$, where $X=cacbab$ in $w_{\{a,b,c\}}$. 
    \\\textbf{Case 5}: If $w_{\{a,b,c\}}=bcas_1$, then $s_1$ needs to begin with $b$ or $c$; otherwise, starting with $a$ would remove the alternation of $a$ with $b$ and $c$. If $s_1$ starts with $b$, then the next letter must be $c$; otherwise, the occurrence of $a$ would remove the alternation between $a$ and $c$. However, $bcabca$ forms a square. Therefore, $s_1$ starts with $c$ and is followed by $ba$. Thus, if $w_{\{a,b,c\}}=bcacbas_2$, then the factor $acb$ occurs in $w_{\{a,b,c\}}$. Therefore, it follows the same condition as Case 2, except that the last occurrences of $b$ and $c$ are removed because $b$ and $c$ already occur as the first and second letters in $w_{\{a,b,c\}}$, respectively. Then $w_{\{a,b,c\}}=bcacbabcacba$, but there exists a square $XX$, where $X=bcacba$ in $w_{\{a,b,c\}}$. 
    \\\textbf{Case 6}: If $w_{\{a,b,c\}}=cbas_1$, then $s_1$ needs to start with $b$ or $c$; otherwise, starting with $a$ would remove the alternation of $a$ with $b$ and $c$. If $s_1$ starts with $c$, then the next letter must be $b$; otherwise, the occurrence of $a$ would remove the alternation between $a$ and $b$. However, $cbacba$ forms a square. Therefore, $s_1$ starts with $b$ and is followed by $ca$. Thus, if $w_{\{a,b,c\}}=cbabcas_2$, then the factor $abc$ occurs in $w_{\{a,b,c\}}$. Therefore, it follows the same condition as Case 1, except that the last occurrences of $b$ and $c$ are removed because $c$ and $b$ already occur as the first and second letters in $w_{\{a,b,c\}}$, respectively. Then $w_{\{a,b,c\}}=cbabcacbabca$, but there exists a square $XX$, where $X=cbabca$ in $w_{\{a,b,c\}}$. 
    
    In all of the cases, we obtain a square $XX$, where $X\in \{a,b,c\}^+$ and $|X|\geq 3$ in $w_{\{a,b,c\}}$. This contradicts the assumption that $w$ is $3$-complete square-free. Therefore, $G$ is not a $3$-complete square-free uniform word-representable graph.
\end{proof}

From Theorem \ref{tm7}, we can obtain an example that contradicts the converse statement of Theorem \ref{lm11}. We know that the crown graph is bipartite, so it avoids $K_3$ as an induced subgraph. According to Theorems \ref{crown1} and \ref{crown2}, the representation number of the crown graph $H_{n,n}$ is $\lceil n/2 \rceil$, $n\geq 5$. Thus, the representation number of the $H_{8,8}$ graph is $\lceil 8/2 \rceil=4$. However, according to Theorem \ref{tm7}, the $H_{8,8}$ graph cannot contain $p$-complete square-free uniform word-representation, when $p=3$. Therefore, there exist word-presentable graphs that are $K_p$-free but are not $p$-complete square-free uniform word-representable. 

According to Theorem \ref{tm7}, it can be directly seen that the graphs with a representation number $\leq 3$ can have the complete square-free uniform representation number $3$. Now, we aim to determine whether it is possible to construct a $3$-complete square-free uniform word-representable graph from an existing $3$-complete square-free uniform word-representable graph.

 If we add an apex vertex to an empty graph, the graph becomes a star graph. Since the star graph is $K_3$-free and $2$-uniform word-representable, according to Theorem \ref{lm12}, the star graph is $3$-complete square-free uniform word-representable. However, if we connect an apex vertex with a non-empty $3$-complete square-free uniform word-representable graph, the resulting graph is no longer $3$-complete square-free uniform word-representable. We prove this statement below.
\begin{cor}\label{lm13}
    Suppose $G(V,E)$ is a non-empty $3$-complete square-free uniform word-representable graph. If the graph $G'$ is obtained by adding an apex vertex to $G$, then $G'$ is not $3$-complete square-free uniform word-representable.
  
\end{cor}
\begin{proof}
    As $G$ is a non-empty graph, let $x\sim y$ where $x,y\in V$. The vertex $v$ is an apex vertex that implies $v\sim x$ and $v\sim y$. So, $vxy$ forms a $K_3$ in the graph $G'$. Therefore, according to the Corollary \ref{cr1}, $G'$ is not $3$-complete square-free uniform word-representable.
\end{proof}

 To create a $3$-complete square-free uniform word, connecting an apex vertex to a $3$-complete square-free uniform word-representable graph does not work. Therefore, we need another operation to create a new $3$-complete square-free uniform word-representable graph. We create a method using the occurrence-based function, and we obtain a word-representation of a new $3$-complete square-free uniform word-representable graph from a known $3$-complete square-free uniform word-representable graph. The definition of an occurrence-based function is described below. 
\begin{dnt}(\textit{\cite{broere2018word}, Definition 3.20.})
    Let $V$ and $V'$ be alphabets, and let $N_k=\{1,\ldots,k\}$. Then $H: V^*\rightarrow (V\times N_k)^*$ is the labelling function of finite words over $V$, where the $i^{th}$ occurrence of each letter $x$ is mapped to the pair $(x, i)$, and $k$ satisfies the property that every symbol occurs at most $k$ times in $w$. An occurrence-based function is a composition $(h o H)$ of a homomorphism $h:(V\times N_k)^*\rightarrow (V')^*$ and the labelling function $H$. 
\end{dnt}

\begin{exm}
    The final permutation $\sigma(w)$ of the $3$-uniform word $w=68314521783$ $6724568314572$ can be defined using the following homomorphism $h$ and the labelling function $H:\{1,2,\ldots,8\}^*\rightarrow (\{1,2,\ldots,8\}\times \{1,2,3\} )^*$, where the $i^{th}$ occurrence of each letter $x$ is mapped to the pair $(x, i)$ and $1\leq i\leq 3$.
        \begin{equation*}
      h(x,i)=
    \begin{cases}
        \epsilon,  & \text{if $i< 3$, }\\
        x, & \text{if $i= 3$. } 
    \end{cases}  
    \end{equation*}
    So, $h(H(683145217836724568314572))$ $=68314572$ $= \sigma(w)$.
\end{exm}
Suppose $w$ is a $3$-complete square-free uniform word that represents the graph $G(V,E)$. By applying an occurrence-based function to $w$, we can obtain a $3$-complete square-free uniform word that represents the graph $G'(V\cup \{v\}, E\cup \{v\sim u: u\in N_x\})$ where $N_x$ is the set that contains the neighbours of $x$ in the graph $G$. 

\begin{theorem}\label{lm14}
    For a $3$-complete square-free uniform word-representable graph $G(V,E)$, if we connect a new vertex $v$ with the neighbours of $x$, where $x\in V$, then the graph $G'(V',E')$ is also $3$-complete square-free uniform word-representable, where $V'=V\cup \{v\}$, $E'=E\cup \{v\sim u: u\in N_x\}$ and $N_x$ is the set that contains the neighbours of $x$.
\end{theorem}
\begin{proof}
    Let $w$ be a $3$-complete square-free uniform word representing $G$. Also, according to Theorem \ref{tm7}, $w$ is at most $3$-uniform word, otherwise $w$ would not be a $3$-complete square-free uniform word. Now, we use the following homomorphism in the occurrence-based function. 
     \begin{equation*}
      h(y,i)=
    \begin{cases}
        y, & \text{if $y\neq x$ }\\
        xv, & \text{if $y=x$, $i$ is odd }  \\
        vx, & \text{if $y=x$, $i$ is even} 
    \end{cases}  
    \end{equation*}
    Now we construct $w'=h(H(w))$, where $H: V^*\rightarrow (V\times \{1,2,3\})^*$ is the labelling function of finite words over $V$, mapping the $i^{th}$ occurrence of each letter $x$ to the pair $(x, i)$, $1\leq i\leq 3$. In the word $w'$, each letter of $w$ is the same except $x$ is replaced by $xv$ in odd positions and $vx$ in even positions. we claim that $w'$ is a $3$-complete square-free uniform word-representation of the graph $G'$. To check whether the graph $G'$ is represented by $w'$ or not, we consider the following three cases:\\
    \textbf{Case 1}: Since, $x$ and $v$ are not adjacent in $G'$, we need to check whether $x$ and $v$ are alternate in $w'$. As subword $xvvx$ is present $w'_{\{x,v\}}$, $x$ and $v$ do not alternate in $w'$.
    \\\textbf{Case 2}: For all $u\in N_x$, $u$ and $v$ are adjacent in the graph $G'$. Without loss of generality, we assume $w_{\{x\}\cup N_x}= xP_1(N_x)xP_2(N_x)xP_3(N_x)$, where $P_i(N_x)$, $1\leq i\leq 3$ are the permutations of $N_x$. Since $x$ is replaced by $xv$ and $vx$ in odd and even occurrence, the occurrences of $x$, $v$ and $u\in N_x$ are $w_{\{x,v\}\cup N_x}=xvP_1(N_x)vxP_2(N_x)xvP_3(N_x)$. So, all vertices of $N_x$ and $v$ are alternating in $w'$.
    \\\textbf{Case 3}: For all $u\in V\setminus (\{x\}\cup N_x)$, $u$ and $x$ are not adjacent in $G$ and $G'$. In the word $w$, $u$ and $x$ do not alternate. Since $x$ is replaced by $xv$ and $vx$ in $w'$, $v$ and $u$ also do not alternate in $w'$.
    
    Therefore, $w'$ represents the graph $G'$. Now, we prove that it is also $3$-complete square-free uniform word. Suppose $w'$ is not a $3$-complete square-free uniform word. It has a square $XX$, where $X\in (V')^+$ and $|X|\geq 3$, in a subword restricted to some vertices. Let $S$ be the set of all vertices present in the square $XX$. As the word $w$ is $3$-complete square-free uniform word, in $w'$, the square has to include $v$; otherwise, $w$ cannot be a $3$-complete square-free uniform word. Therefore, $X=X_1vX_2$, where $X_1$ and $X_2$ contain the vertices of $S$. But using the construction, we can say that $w_{S\cup \{x\}}$ contains $X_1xvX_2X_1vxX_2$ or $X_1vxX_2X_1xvX_2$ as a factor. In either of the cases removing $v$ yields a square $X_1xX_2X_1xX_2$ where $|X_1xX_2|=|X_1vX_2|\geq 3$. However, it contradicts that $w$ is a $3$-complete square-free uniform word. Therefore, $w'$ is a $3$-complete square-free uniform word.
\end{proof}

It is well known that cycle graphs are $2$-word-representable. In this example, we apply the construction described in Theorem~\ref{lm14} to obtain a $3$-uniform square-free word representation of an extended graph.

\begin{exm}
The word $w = 5213243541$ is a $2$-uniform word representing the cycle graph $C_5$. Using the method of Theorem~\ref{lm14}, we construct a $3$-complete square-free uniform word representing the new graph $G_1$ shown in Figure \ref{fig2}. The graph $G_1$ is obtained from $C_5$ by adding a new vertex $1'$, which is adjacent to all neighbours of the vertex $1$ in $C_5$. To construct the corresponding word, we define the following homomorphism:
 \begin{equation*}
      h(y,i)=
    \begin{cases}
        y, & \text{if $y\neq 1$ }\\
        11', & \text{if $y=1$, $i=1$ }  \\
        1'1, & \text{if $y=1$, $i=2$} 
    \end{cases}  
    \end{equation*}
The word $w_1 =h(H(w))= 5211'3243541'1$, where $H:\{1,2,3,4,5\}^*\rightarrow (\{1,2,3,4,5\}\times \{1,2\})^*$ is a $3$-complete square-free uniform word representing the graph $G_1$.
We can again construct another $3$-complete square-free uniform word-representable graph $G_2$ shown in Figure \ref{fig2}. The graph $G_2$ is obtained from $G_1$ by adding a new vertex $5'$, which is adjacent to all neighbours of the vertex $5$ in $G_1$. We used the following homomorphism to obtain the $3$-complete square-free uniform word-representation:
 \begin{equation*}
      h(y,i)=
    \begin{cases}
        y, & \text{if $y\neq 5$ }\\
        55', & \text{if $y=5$, $i=1$ }  \\
        5'5, & \text{if $y=5$, $i=2$} 
    \end{cases}  
    \end{equation*}

The word $w_2 =h(H(w_1))= 55'211'32435'541'1$,  where $H:\{1,2,3,4,5,1'\}^*\rightarrow (\{1,2,3,4,5,1'\}\times \{1,2\})^*$ is a $3$-complete square-free uniform word representing the graph $G_2$.

    \begin{figure}[h]
\begin{center}
\begin{tabular}{c c c}
   
\begin{tikzpicture}[node distance=1cm,auto,main node/.style={circle,draw,inner sep=1pt,minimum size=5pt}]

\node[main node] (1) {1};
\node[main node] (2) [below right of=1] {2};
\node[main node] (3) [below of=2] {3};
\node[main node] (5) [below left of=1] {5};
\node[main node] (4) [below of=5] {4};
\node (6) [below right of=4] {$C_5$};

\path
(1) edge (2)
(2) edge (3)
(3) edge (4)
(4) edge (5)
(1) edge (5);

\end{tikzpicture}
  \   \  \ & 

     \begin{tikzpicture}[node distance=1cm,auto,main node/.style={circle,draw,inner sep=1pt,minimum size=5pt}]

\node[main node] (1) {1};
\node[main node] (2) [below right of=1] {2};
\node[main node] (3) [below of=2] {3};
\node[main node] (5) [below left of=1] {5};
\node[main node] (4) [below of=5] {4};
\node (6) [below right of=4] {$G_1$};
\node[main node] (7) [below of=1] {$1'$};

\path
(1) edge (2)
(2) edge (3)
(3) edge (4)
(4) edge (5)
(1) edge (5)
(7) edge (2)
(7) edge (5);

\end{tikzpicture}

\  \ \ &
\begin{tikzpicture}[node distance=1cm,auto,main node/.style={circle,draw,inner sep=1pt,minimum size=5pt}]

\node[main node] (1) {1};
\node[main node] (2) [below right of=1] {2};
\node[main node] (3) [below of=2] {3};
\node[main node] (5) [below left of=1] {5};
\node[main node] (4) [below of=5] {4};
\node (6) [below right of=4] {$G_2$};
\node[main node] (7) [below of=1] {$1'$};
\node[main node] (8) [left of=5] {$5'$};

\path
(1) edge (2)
(2) edge (3)
(3) edge (4)
(4) edge (5)
(1) edge (5)
(7) edge (2)
(7) edge (5)
(8) edge (1)
(8) edge (4)
(8) edge (7);

\end{tikzpicture}
\end{tabular}
\caption{\label{fig2} Examples of $3$-complete square-free uniform word-representable graphs}
 
\end{center}
\end{figure}
\end{exm}
We prove that word-representable graphs $G$ with representation numbers $> 3$ do not have 3-complete square-free uniform words. It remains open whether all $K_3$-free graphs with representation number $3$ are $3$-complete square-free uniform word-representable.
\section{Conclusion}
 We introduced the notion of a $p$-complete square-free uniform word-representation and derived several of its fundamental properties. We described a constructive process for obtaining a $(p+1)$-complete square-free uniform word-representable graph from a $p$-complete square-free uniform word-representable graph. We also showed that graphs with representation number at most $3$ can have a $3$-complete square-free uniform word-representation. We conclude by listing several open problems and possible directions for further research related to these topics.

\begin{enumerate}
    \item Characterise the word-representable graphs with representation number at least $3$ that admit $p$-complete square-free uniform word-representations for $p>3$.
    \item Determine the word-representable graphs with representation number at least $3$ that do not admit any $p$-complete square-free uniform word-representation for $p>3$.
    \item Characterise $p$-complete square-free uniform word-representable graphs.
    \item Since every $K_3$-free $2$-representable graph is $3$-complete square-free uniformly word-representable, it is natural to study the enumeration of $3$-complete square-free uniform word-representations of $2$-representable $K_3$-free graphs. 
    \item We showed that $3$-complete square-free uniform word-representable graphs cannot have representation number greater than $3$. It remains open whether a similar result holds for $p$-complete square-free uniform word-representable graphs with $p>3$. In particular, does there exist a word-representable graph $G$ whose representation number is greater than $p$, while its complete square-free uniform representation number is at most $p$?
\end{enumerate}

	\bibliographystyle{plain}
	\bibliography{ref.bib}
\end{document}